\documentclass{llncs}

\usepackage{amsmath}
\usepackage{graphicx}
\usepackage[space]{grffile}
\usepackage{latexsym,fontenc}
\usepackage{textcomp}
\usepackage{longtable}
\usepackage{multirow,booktabs}
\usepackage{amsfonts,amsmath,amssymb,stmaryrd}
\usepackage{url}
\usepackage{comment}
\usepackage{hyperref}

\usepackage{enumerate}

\usepackage{bm}
\usepackage{xcolor}
\usepackage{microtype}

\usepackage[normalem]{ulem}
\usepackage{tikz}

\usetikzlibrary{arrows}
\usetikzlibrary{positioning}

\newenvironment{myquote}{\list{}{\leftmargin=0.1in\rightmargin=0.1in\topsep=3pt}\item[]}{\endlist}

\newtheorem{observation}{Observation}

\newcommand{\C}{\mathfrak{C}}
\newcommand{\B}{\mathsf{B}}
\newcommand{\NB}{\mathsf{NB}}
\newcommand{\I}{I}
\newcommand{\g}{\mathcal{G}}

\newcommand{\N}{N}
\newcommand{\T}{\mathcal T}

\renewcommand{\O}{{\bf o}}

\newcommand{\Q}{{\bf Q}}

\newcommand{\Win}{\mathcal C}

\newcommand{\tuple}[1]{\left\langle #1 \right\rangle}
\newcommand{\set}[1]{\left\{ #1 \right\}}

\renewcommand{\phi}{\varphi}

\renewcommand{\phi}{\varphi}

\makeindex

\title{
Decentralization in Open Quorum Systems
}
\subtitle{Limitative Results for Ripple and Stellar}

\author{Andrea Bracciali$^1$ \and Davide Grossi$^{2,3,4}$ \and Ronald de Haan$^3$}
\institute{$^1$Department of Computer Science \\ University of Stirling  \\
$^2$Bernoulli Institute for Mathematics, Computer Science and Artificial Intelligence \\ University of Groningen \\
$^3$Institute for Logic, Language and Computation \\ University of Amsterdam \\
$^4$Amsterdam Center for Law and Economics \\
University of Amsterdam
}
\begin{document}
\maketitle

\begin{abstract}
\sloppy
Decentralisation is one of the promises introduced by blockchain technologies: fair and secure interaction amongst peers with no dominant positions, single points of failure or censorship. Decentralisation, however, appears difficult to be formally defined, possibly a continuum property of systems that can be more or less decentralised, or can tend to decentralisation in their lifetime.
In this paper we focus on decentralisation in quorum-based approaches to open (permissionless) consensus as illustrated in influential protocols such as the Ripple and Stellar protocols. Drawing from game theory and computational complexity, we establish limiting results concerning the decentralisation vs. safety trade-off in Ripple and Stellar, and we propose a novel methodology to formalise and quantitatively analyse decentralisation in this type of blockchains.
\end{abstract}

\begin{keywords}
Decentralisation analysis, Quorum systems, Trust Networks, Game theory, Power indices, Ripple, Stellar
\end{keywords}


\section{Introduction}
\label{s.deccons}

To allow ``any two willing parties to transact directly with each other without the need for a trusted third party"~ \cite{bitcoin} was one of the main motivations for the introduction of the Bitcoin blockchain, and several earlier attempts at digital currencies.  A blockchain is a distributed state machine in charge of guaranteeing the correctness and trustability of data,%
\footnote{Blockchains can also make the computation trustable, e.g., guaranteeing the fair and untamperable execution of agreements among peers encoded as programs.}
 e.g., monetary transactions in the case of Bitcoin. State updates are recorded in a chain of data blocks. Data are protected by replication of the state, i.e., of the chain of blocks, within a network of peers. The blockchain protocol must guarantee some form of distributed consensus allowing peers to agree on the information contained in the blockchain, e.g., who has been paid, and that no double spending of virtual coins has occurred, without the supervision of a centralised authority -- a currency without a central bank.

\subsubsection{Context: the quest for decentralisation}
Decentralisation, which can be informally understood as the lack of dominant positions amongst the independent and untrusted peers, is then a strongly desirable property of blockchains. 
Decentralisation may involve several aspects of a blockchain definition and architecture, and is strongly connected to the problem of governance. There can be centralisation in the software maintenance, the management of the community of peers and users and associated access policies, the management of the tokenomics and incentives\footnote{Monetary and incentive policies, e.g., the economics of tokens,  the motivational reward for peers, and the initial distribution of stake in PoS (see below), may affect the stability and long-term survival of the blockchain.},  the resolution of disputes, and the pooling of peers (see~\cite{Gervais2014BTCcentralised} for early considerations on the decentralisation in Bitcoin). 

In this paper we address decentralisation in distributed consensus. Decentralisation in the sense of  ``any two willing parties" somehow implies openness, embodied by  {\em permissionsless} blockchains, where participation is allowed in a generally unrestricted way. Permissionless blockchains are clearly exposed to the presence of Byzantine peers, i.e.,  dishonest peers trying to exploit the network and not bound to the blockchain protocol. Byzantine distributed consensus is a long-standing problem, from Lamport's characterisation~\cite{lamport82byzantine} and the FLP impossibility result~\cite{Fischer1985impossibility}, to the subsequent research on data replication and consistency based on Byzantine Fault-tolerant consensus (BFT),~\cite{malkhi98byzantine,schneider90implementing}. Several proposals are  currently competing in a multi-billion market, addressing the so-called {\em blockchain trilemma}, i.e., achieving {\em security}, {\em scalability} and {\em decentralisation} together.

One of the breakthroughs of Bitcoin was the introduction of the Proof-of-Work (PoW)~\cite{PoW93} as a mechanism to enable a probabilistic Byzantine distributed consensus. Informally speaking, by solving a computationally hard problem one of the peers is entitled to create the next block, cryptographically linked to the previous ones. Under the assumption that Byzantine computational power is suitably limited within the network, the probability that enough work can be channeled to alter block history decreases with the ageing of the blocks, as much as new blocks are created~\cite{Garay2015,Eyal2018majority}. Bitcoin reaches finality with an acceptable probability in about one hour (6 blocks),  with limited transactions per second.\footnote{More technical and comprehensive introductions to blockchains can be found in~\cite{bonneau2015sok,narayanan16handbook,wattenhofer17distributed}.}

In  {\em Proof-Of-Stake}  (PoS) blockchains peers contribute to the definition of the next block with a probability proportional to the stake (coins), rather than computational power, they detain in the system. Safety follows from the honest peers detaining the majority of stake. Scalability improves in Proof-of-Stake, but the management  of security typically results in being more complex.

The BFT paradigm has also been proposed for blockchain consensus, providing scalability in transaction per second thanks to low transaction latency and high throughput. BFT, however, is more constrained in terms of the scalability in the number of peers~\cite{vukolic15quest}, since the number and identity of peers needs to be known and in some cases fixed~\cite{Castro:1999:PBFT}. This kind of blockchain has been proposed, for instance, for financial services, where a limited number of known and certified peers need to exchange fast and numerous transactions. It is worth remarking that if consensus requires control on peers, a centralised authority might be required, with implications also on identity, privacy and censorship.

\subsubsection{Research question}

In this paper we focus on BFT blockchains based on quorum systems \cite{vukolic12quorum}. In such systems consensus emerges from neighbourhoods of peers and the properties of such neighborhoods, together with assumptions on Byzantine failure thresholds, are essential to guarantee the liveness and safety of the consensus protocol, that is, whether honest peers are able to eventually reach consensus on a correct next state. At the heart of these protocols is a notion of trust between peers: nodes select which other nodes to trust, and listen to, in the network. Interested in understanding how much decentralisation can (or cannot) be achieved in principle in such systems, we address the following question:
\begin{center}
{\em To what extent can consensus be decentralized, when based on trust networks?}
\end{center}
We approach the question by using tools from cooperative game theory (specifically the theory of command games \cite{hu03authority,hu03authorityB}) and computational complexity theory. The interface of methods from theoretical economics and computational complexity have proven extremely prolific in other areas of computer science and artificial intelligence, such as computational social choice theory \cite{comsoc_handbook,grandi18voting}. Our paper also aims to showcase these methods for general investigations on blockchain consensus and decentralisation. 

We will consider Ripple \cite{chase18analysis,schwartz14ripple} and Stellar~\cite{mazieres16stellar}, two  Quorum-based blockchains attempting to extend the applicability of the BFT paradigm from a permissioned to a permissionless setting, aiming at improving decentralisation. 

Ripple provides frictionless global payments and corporate-oriented efficient transactions. It currently relies on a list of ``authorised'' validators\footnote{At the time of writing the list consists of about 30 validators, available at \url{https://xrpcharts.ripple.com}}
in charge of the correctness of transactions. Access is permissioned and each peer will need to have in their neighbourhood of trust a number of validators from the list. While the list was originally entirely composed by Ripple validators, today third-party validators, e.g., private companies and universities, have been included.

Stellar  provides payments and asset management to corporate and individuals, and aims to push decentralisation further by offering open membership and allowing peers to autonomously define their trust networks, i.e., the set of validators that they trust. However, strong constraints hold on the topology of such trust networks.

Both Ripple and Stellar have been object of criticism with respect to the level of decentralisation of their current implementations, and the need for further research on protocols like Ripple and Stellar is emphasised, for instance, in \cite{cachin17blockchain}. 


\subsubsection{Related work}

Even though Ripple and Stellar are, respectively, the third and tenth blockchain systems in terms of market capitalization, very little foundational work exists on their protocols. Correctness analyses of Ripple have been proposed in \cite{chase18analysis}, and of Stellar in \cite{mazieres16stellar,garcaprez2018OPODISStellar}.
A specific study on the issue of decentralisation in Stellar has also very recently been presented in~\cite{kim2019IEEEPSB}.  Authors investigate the current topology of Stellar's quorum slices by means of an extended version of PageRank that they introduce to evaluate nodes' influence. Findings about the current status show centralisation on two critical validators, which are controlled by the Stellar Foundation. 

Our paper contributes further general results on the level of decentralization that could reasonably be achieved in consensus based on open quorum systems.

\subsubsection{Paper contribution and outline}
The contributions of this paper are: 
\begin{itemize}

\item a novel theoretical framework, rooted in economic theory (command games \cite{hu03authority,hu03authorityB}, {\em power indices}~\cite{penrose46elementary,banzhaf65weighted}), to ascertain the influence that peers can exert on each other in quorum systems based on trust networks. This contributes a novel methodology for a much needed quantitative evaluation of decentralisation in blockchain (here in the context of consensus). The proposed methods are applied to Ripple and Stellar (Theorem \ref{prop:Ireg}).


\item a general impossibility of decentralisation result for a class of consensus protocols of the Ripple type (Theorem \ref{th:trust}), which are based on trust networks with a fixed threshold of tolerable Byzantine peers.
 This results implies that in Ripple the necessary existence of validators that must be trusted by every peer in the network, hindering the possibility of full decentralisation.
 
 \item an appraisal of computational barriers to decentralization in protocols like Stellar, based on so-called federated Byzantine agreement systems. Specifically, we show that constraints that are necessary to guarantee the safety of the network require peers to be able to solve problems that are computationally intractable in principle (Theorems \ref{prop:qi-conp-complete} and \ref{prop:sliceadd-conp-complete}).
 This result imposes in Stellar limitations to the autonomous construction of trust networks by peers, limiting also in this case the possibility of full decentralisation.

\end{itemize}
Trust networks and command games are introduced in Section~\ref{sec:prelim}, impossibility and intractability results are presented in Section~\ref{s.impossibility}, and decentralisation measures in  Section~\ref{sec:influence}.
%
%
Section~\ref{sect:concl} contains final considerations and future outlooks. Relevant literature is discussed throughout the paper, proofs appear in the Appendix.


\section{Preliminaries} 
\label{sec:prelim}

In this section we link the open quorum systems underlying Ripple ({\em trust networks} \cite{ghosh07mechanism}) and Stellar ({\em federated Byzantine agreement systems}, FBAS, \cite{mazieres16stellar}) to structures studied in the economic theory literature, known as {\em command games}. This will then allow us to import concepts and results from the field of command games to the workings of consensus in Ripple and Stellar. 

\subsection{Byzantine Trust Networks}

A set of peers, hereafter nodes, want to get to an agreement on a binary opinion $x \in \set{0,1}$, e.g., whether a given transaction should belong, or not, to the current ledger. The set is open, but we consider a snapshot at a given point in time.  Byzantine nodes, differently from honest ones, can hold and reveal multiple, inconsistent opinions. The goal of consensus is to have all honest nodes eventually agreeing on the same opinion (no-forking). Crucially, a honest node's opinion depends on the opinions revealed by the nodes (honest or byzantine) that it trusts.

\begin{definition}
\label{def:btn}
A {\em Byzantine trust network} (BTN) is a tuple $\T = \tuple{N, H, T_i, \Win_i}$ where:
\begin{itemize}
\item $\N = \set{1, \ldots, n}$ is a finite set of nodes.
\item $H \subseteq N$ is the set of honest nodes. $B = N \backslash H$ is the set of Byzantine nodes.  
\item $T_i \subseteq N$, for each node $i \in H$, is the non-empty set of nodes that $i$ trusts, i.e., its {\em trust set}.\footnote{Ripple and Stellar refer to trust sets as unique node lists (UNLs).} 
\item $\Win_i \subset 2^{T_i}$, for each honest node $i \in H$, is the collection of sets of nodes, among those that $i$ trusts, that can determine $i$'s opinion. We refer to $\Win_i$ as the set of {\em winning coalitions} for node $i$. For ease of presentation we will sometimes treat $\Win$ also as a function assigning a set of subsets of $\N$ to each honest node.
\end{itemize}
We will sometimes assume that, for all $i$,  $i\in T_i$.
We will sometimes furthermore assume that for all $i \in H$, $\set{i} \in \Win_i$. In such a case the BTN is said to be {\em vetoed}.
\end{definition}
Intuitively, a winning coalition $C \in \Win_i$ is a set of nodes such that, if all members of one of such coalition agree on a value, then that value is also $i$'s opinion. When $i$ belongs to $\Win_i$, $i$ cannot validate an opinion unless it also holds such opinion (it holds a veto for its own validation). 
In the Stellar white paper \cite{mazieres16stellar} BTN are referred to as {\em federated Byzantine agreement systems} (FBAS), or as {\em federated Byzantine quorum systems} in \cite{garcaprez2018OPODISStellar},
and the winning coalitions of a node are referred to as {\em quorum slices}. BTNs are also known structures in the economic theory literature, where they are referred to as {\em command games} \cite{hu03authority,hu03authorityB}, or as {\em simple game structures} \cite{karos15indirect}.

\medskip

A natural class of BTNs is obtained by associating a quota, or threshold, $q_i \in (0.5, 1]$ to each honest node $i$:\footnote{Cf. \cite{garcia-molina85how}.}
\begin{definition}
\label{def:qbtn}
A Quota Byzantine Trust Network (QBTN) is a BTN such that for all $i \in H$ there exists a quota $q_i \in (0.5, 1]$ such that:
$$
\Win_i = \set{C \subseteq T_i \mid |C| \geq q_i \cdot | T_i |}.
$$
A QBTN is therefore denoted by a tuple $\T = \tuple{N, H, T_i, q_i}$. A QBTN is said to be {\em uniform} whenever $q_i = q_j$ for any $i,j \in H$.
\end{definition}
Intuitively, QBTNs are BTNs where the winning coalitions of a node are determined by a numerical quota: $i$'s opinion is determined whenever at least $q_i$ nodes in $T_i$ hold that opinion. Assuming for each $T_i$ a standard failure model with a fraction of Byzantine node $b_i \in [0, 0.25)$ (cf. \cite{lamport82byzantine}), when $q_i = 1 - b_i$ it is guaranteed that {\em i)} the quota is met whenever the honest nodes in $T_i$ agree, and {\em ii)} if the quota is met for an opinion $x \in \set{0,1}$, then there is at least an honest majority of nodes with opinion $x$ in the trust set of $i$. So in this paper we will assume quota fall in the $[0.75, 1]$ range. A QBTS is said to be {\em uniform} whenever $q_i = q_j$ for any $i,j \in H$.

\medskip 

The Ripple consensus protocol  \cite{chase18analysis,schwartz14ripple} is based on uniform QBTNs with quotas set to $0.8$. 
The Stellar consensus protocol as described in \cite{mazieres16stellar} is not based on quota but requires the generality of BTNs while assuming them to be vetoed.

\begin{remark} \label{rem:dictators}
In Definition \ref{def:btn} we associate winning coalitions only to honest nodes. We do this for simplicity but it should be clear that trivial collections of winning coalitions can be associate also to Byzantine nodes. Since the opinion of a Byzantine node $i$ is not influenced by any other node its trivial collection of winning coalitions is the set $\set{ C \subseteq N \mid \set{i} \subseteq C}$, that is, the set of all coalitions containing $i$. Intuitively, this amounts to stating that $i$ is the only node influencing its own opinion.
\end{remark}

\subsection{Opinions and Safety}

At any given time, the collection of each node's opinions defines an opinion profile that associates a `genuine' opinion from $\set{0,1}$ to every honest node (the opinion that the node reveals to the network). And to each Byzantine node it associates a function from honest nodes to opinions. This function represents the values that each Byzantine node would reveal to each honest node in the network that includes it in its trust set.
\begin{definition} 
An {\em opinion profile} $\O: N \to \set{0,1} \cup \set{0,1}^H$ such that $\O(i) \in \set{0,1}$ if $i \in H$ and $\O(i) \in \set{0,1}^H$ if $i \in B$.
\end{definition}

\noindent
Intuitively, $i$'s opinion is settled whenever a winning coalition of trusted nodes holds that opinion.
In a QBTN, $i$'s opinion (if $i$ is honest) is settled whenever there are at least $q_i \cdot |T_i|$ nodes with a same opinion among the nodes it trusts. Given an opinion profile $\O$ we denote by 
\begin{align}
& T^\O_j(x) = \set{i \in T_j \mid \O(i) = x \mbox{    if    } i \in H, \mbox{    or    } \O(i)(j) = x \mbox{    if    } i \in B}.
\end{align} 
This is the set of nodes (honest or Byzantine), among those that $j$ trusts, holding opinion $x$ in profile $\O$. 
We say that $i$ {\em validates} $x \in \set{0,1}$ (in a given $\O$) if $T^\O_i(x) \in \Win_i$. We then say that an opinion profile $\O$ is {\em forked} (or, is a fork) if there are two honest nodes $i,j \in H$ such that $i$ validates $x$ and $j$ validates $\overline{x}$ in $\O$, where, given $x \in \set{0,1}$,  $\overline{x}$ represents the element of singleton $\set{0,1}\backslash \set{x}$. 

\medskip

Let us introduce some auxiliary notions. Let $s: N \to 2^N$ be a function picking, for any agent $i$, one coalition out of $\Win_i$.\footnote{Clearly there are $|\prod_{i \in H}Win_i|$ such functions.} Each function $s$ induces an operator $F_s: 2^N \to 2^N$ such that $F_s(C) = \bigcup_{i \in C} s(i)$, collecting, for each $i$ in $C$, the winning coalition $s(i)$ picked by function $s$. We further denote by $F^k_s$ the $k$-th iteration of $F_s$. Since $\N$ is finite, for any $C \subseteq \N$ there exists $k \in \mathbb{N}$ such that $F^k_s(C) = F^{k+1}_s(C)$.\footnote{Cf. the Knaster-Tarski theorem \cite{davey90introduction}.}

We say that an opinion profile $\O$ is {\em strongly forked} (or, is a strong fork) if there are two honest nodes $i,j \in H$ and a function $s$ such that all nodes in $\bigcup_{1 \leq m} F_s^m(\set{i})$ agree on $x$ and all nodes in $\bigcup_{1 \leq m} F_s^m(\set{j})$ agree on $\overline{x}$. Formally: for all $k \in \bigcup_{1 \leq m} F_s^m(\set{i})$ $\O(k) = x$, and for all $k \in \bigcup_{1 \leq m} F_s^m(\set{j})$ $\O(k) = \overline{x}$. 
That is, there is a winning coalition for $i$ agreeing on $x$ and a winning coalition for $j$ agreeing on $\overline{x}$, {\em and} all nodes in that winning coalition for $i$ also have a winning coalition agreeing on $x$ and all nodes in that winning coalition for $j$ have a winning coalition agreeing on $\overline{x}$, {\em and} so on.

\begin{definition} 
\label{def:safety}
A BTN is {\em safe} if there exists no forked profile for it. It is {\em weakly safe} if there exists no strongly forked profile for it.
\end{definition}
Safety rules out the possibility that two honest nodes may settle on different opinions. Weak safety allows for forks of only a limited kind. It rules out the possibility that forks are of a `deep' kind involving all winning coalitions upon which the diverging opinions are rooted. 
Clearly safety implies weak safety but not vice versa.

\section{(De)centralisation and (In)tractability}
\label{s.impossibility}

This section explores inherent limitations present in the above notion of safety for BTNs, establishing general limitative results for the class of consensus protocols based on them, such as Ripple and Stellar. First, it focuses on uniform QBTNs (Definition \ref{def:qbtn}), as exemplified by the Ripple consensus protocol, showing that safety drastically limits the freedom of nodes in selecting trustees. Second, it focuses on safety for general BTNs (Definition \ref{def:btn}), as exemplified by the Stellar consensus protocol, showing that, even though safety in such settings allows for more freedom on the part of nodes, it does require single nodes to solve decision problems that are, in principle, computationally intractable.


\subsection{Safety Implies Centralization in Uniform QBTNs}

We show that for safe consensus to be possible on uniform QBTNs nodes cannot be fully free to choose their trust set. 

The result builds on ideas from \cite{chase18analysis}. We will use the following auxiliary definition, also borrowed from \cite{chase18analysis}:
\begin{align}
\beta_{ij} = \min\set{|T_i \cap T_j|, b_i \cdot |T_i|, b_j \cdot |T_j|}. \label{eq:beta}
\end{align}
Intuitively, $\beta_{ij}$ denotes the number of Byzantine agents present in the intersection of the trust sets of $i$ and $j$. Such a number equals the maximum amount of Byzantine nodes assumed by the node, either $i$ or $j$, that tolerates fewer Byzantine nodes. However, such a number cannot obviously exceed the size of the intersection itself.

\begin{lemma}[\cite{chase18analysis}] \label{lemma:most}
Let $\tuple{N, H, T_i, q_i}$ be a QBTN. For any profile $\O$ and node $i \in H$, if $|T_i^\O(x)| > 0$ then for any $j \in H$, 
\begin{align}
|T_j^\O(x) \cap H| & \geq |T_i \cap T_j| + |T_i^\O(x)| - |T_i| - \beta_{ij} \label{uno} \\
|T_j^\O(\overline{x})| & \leq |T_j| -  |T_i \cap T_j| - |T_i^\O(x)| + |T_i| + \beta_{ij} \label{due}
\end{align}
\end{lemma}

\noindent
This lemma establishes a lower bound on the number of honest nodes with opinion $x$ that a honest node $j$ can observe in its trust set, given the number of nodes (not necessarily honest) that another honest node $i$ observes. It is used in the proof of Lemma~\ref{lemma:half}. Notice that the upper bound in \eqref{due} is not necessarily strict as illustrated in the following example.
\begin{example}
Let $\tuple{N, H, T_i, q_i}$ be such that: $N = \set{1, \ldots, 5}$, $B = \set{3}$ (recall $N = H \cup B$), $T_1 = T_2 = \set{1, 2, 3}$ and $T_4 = T_5 = \set{3, 4, 5}$, $q_1 = q_2 = q_4 = q_5 = 1$. Let then $\O$ be such that $\O(1) = \O(2) = 1$, $\O(4) = \O(5) = 0$, finally $\O(3)$ be such that $\O(3)(1) = \O(3)(2) = 1$ and $\O(3)(4) = \O(3)(5) = 0$. We have that $4$ and $5$ see no honest node with opinion $1$, and we thus have $|T_4^\O(0)| = |T_5^\O(0)| = 3$.
\end{example}

\begin{lemma} \label{lemma:half}
Let $\T = \tuple{N, H, T_i, q_i}$ be a safe uniform QBTN. Then for all $i,j \in H$:
$$
|T_i \cap T_j| > \frac{b}{1-b} \cdot (|T_i| + |T_j|).
$$ 
\end{lemma}
It is worth observing that Lemma~\ref{lemma:half} establishes a conservative lower bound on the size of the intersection of trust sets required by safety. It is not difficult to construct uniform QBTNs where the intersection equals $\frac{1}{2}(|T_i| + |T_j|)$. This is the case, for instance, of BTNs where $|T_i| = |T_j|$ for all $i,j \in H$.

\begin{lemma} \label{lemma:intersect}
Let $\tuple{N, H, T_i, q_i}$ be a uniform BTN. If for all $i,j \in H$, 
$$|T_i \cap T_j| > 0.25 \cdot (|T_i| + |T_j|),$$
 then 
 $$\bigcap_{i \in H} T_i \neq \emptyset.$$
\end{lemma}

\begin{theorem} \label{th:trust}
In uniform QBTNs with quotas $q \in [0.75, 0.80]$, safety implies the existence of nodes that are trusted by all honest nodes.
\end{theorem}
\begin{proof}
The result follows directly from Lemmas \ref{lemma:half} and \ref{lemma:intersect} and the observation that for $q \in [0.75, 0.80]$ we have $\frac{b}{1-b} \geq 0.25$.\qed
\end{proof}
If we understand decentralisation as the property of trust networks in which nodes have full freedom on whom to trust in the network, then the theorem can be interpreted as a general impossibility result for decentralised consensus based on QBTNs: if quotas are uniform, and set in a reasonable way in order to cope with the presence of Byzantine nodes in trust sets, then the existence of nodes that are trusted by everyone is a necessary condition for the safety of consensus. Furthermore, beyond limiting the choice of nodes, a (limited) set of such nodes clearly represents a dominant position and risk factor for the blockchain.

\medskip

\noindent
In general, Theorem \ref{th:trust} applies to any consensus protocol based on uniform BTNs. In particular, it applies to the Ripple consensus protocol, which uses quota ($q=0.8$).  In a way, Theorem~\ref{th:trust} provides an ex-post analytical justification to the current design of the Ripple trust network where all trust sets are required to include a same subset of nodes (cf. \cite{chase18analysis}). Currently Ripple relies on a single UNL mostly controlled by Ripple, although plans for further decentralisation are under discussion.\footnote{%
Cf. \url{https://xrpcharts.ripple.com/}.
}


\subsection{Safety and Quorum Intersection in BTNs}

In this and the next section we consider the general case of (vetoed) BTNs, to which Theorem \ref{th:trust} does not apply.  This more general setting applies to Stellar as presented in its white paper \cite{mazieres16stellar} where the Stellar consensus protocol does not presuppose uniformity of quotas. Actually, Stellar aims to offer open membership and freedom in choosing it's own trust networks, which, together with BFT good scalability, would yield a decentralised and efficient blockchain, an interesting value proposition. 
In such a setting an intuitive necessary condition for safety is that trust networks are sufficiently `interconnected', in the following sense. 

Let $\tuple{N, H, T_i, \Win_i}$ be a vetoed BTN. A non-empty set $Q \subseteq \N$ is called a {\em quorum} if and only if $\forall i \in Q$, $\exists C \in \Win_i$ s.t. $C \subseteq Q$. Quora are, intuitively, sets of nodes that can form an agreement. Such sets need to be in a pairwise non-empty intersection relation.

\begin{definition}[Quorum intersection \cite{mazieres16stellar}] \label{def:quorum}
A vetoed BTN enjoys quorum intersection (QI) whenever for any two sets $Q_1, Q_2 \subseteq H$, if $Q_1$ and $Q_2$ are quora, then $Q_1 \cap Q_2 \neq \emptyset$.
\end{definition}

\begin{figure}[t]
\begin{center}
\begin{tikzpicture}

\node (1) {$1$};
\node (2) [below left of=1] {$2$};
\node (3) [below right of=1] {$3$};

\node (4) [right=2cm of 1]{$4$};
\node (5) [below left of=4] {$5$};
\node (6) [below right of=4] {$6$};

\draw[<->]
(2) edge (1)
(3) edge (1)
(3) edge (2);

\draw[<->]
(5) edge (4)
(6) edge (4)
(6) edge (5);
\end{tikzpicture}
\end{center}
\caption{
Example from \cite{mazieres16stellar} of a vetoed BTN lacking QI. Arrows denote which nodes each node trusts (reflexive arrows omitted). $\Win(1) = \Win(2) = \Win(3) = \set{\set{1,2,3}}$ and  $\Win(4) = \Win(5) = \Win(6) = \set{\set{4,5,6}}$.
}
\label{fig:M6}
\end{figure}

\begin{example}
  In Figure~\ref{fig:M6} the quora are
  $$
  \{1,2,3\}\quad
  \{4,5,6\}\quad
  \{1,2,3,4,5,6\}
  $$
This command game does not enjoy QI, but both of its disjoint components (with support $\set{1,2,3}$ and $\set{4,5,6}$) do.
Suppose instead that $\set{1,2,3,5} \in \Win(3)$, that is, $3$ also `looks at' $5$ to determine its own value. Then the system would satisfy QI with quora:
$$
\{4,5,6\}\quad
\{1,2,3,4,5,6\}
$$
\end{example}

\begin{example} \label{example:unanimity}
In a BTN $\T$ where $\forall i \in N$, $\Win(i) = \set{N}$, the set $N$ of all nodes is the unique quorum, and $\T$ trivially enjoys quorum intersection.
\end{example}  

In fact, there is a close relationship between quorum intersection and the property of weak safety: 
\begin{theorem}
\label{th:QI}
A vetoed BTN is weakly safe iff any two quora intersect and such intersection contains at least one honest node.
\end{theorem}
Clearly, nodes in a BTN cannot know which nodes are Byzantine so their best effort to guarantee weak safety is to guarantee QI is not violated.\footnote{Although a BTN could be complemented by a failure model consisting of a set sets of possible Byzantine nodes representing the possible failure scenario that nodes should consider (cf. \cite{garcaprez2018OPODISStellar}).}


\subsection{The intractability of maintaining quorum intersection}

Quorum intersection is in fact assumed by all the existing correctness analyses of Stellar \cite{mazieres16stellar,garcaprez2018OPODISStellar}.
It is furthermore stressed in \cite[p. 9]{mazieres16stellar} that: ``[\ldots]
it is the responsibility of each node $i$ to ensure $\Win_i$ [notation adapted] does not violate quorum intersection.''

The key question, from a safety perspective, becomes therefore whether single nodes can reasonably be tasked with maintaining QI. Apart from incentive issues, which have also been flagged \cite{kim2019IEEEPSB}, we argue that this is a problematic requirement from a merely computational standpoint. This might not be an issue in the current, small-scale, Stellar configuration (although an instance of QI failure has been recently reported \cite{lokhava19fast}), but it is something to be considered in a path towards full decentralisation with a full-scale number of nodes and validators. As our analysis below shows, maintaining QI is a computationally intractable problem.

We present two results. First we show that deciding whether a given BTN satisfies QI is intractable.\footnote{An equivalent result has been recently presented in~\cite{Lachowski2019QInp}. That paper provides a proof of NP-completeness (via reduction from the Set Splitting Problem) of the complementary problem for which we prove coNP-completeness (via reduction from 3SAT).}
Second, we show that deciding whether adding a new trust set with winning coalitions preserves QI on a given BTN---arguably the decision problem that nodes need to solve when linking to the Stellar network---is also computationally intractable (again coNP-hard). We argue that these results point to a possible computational bottleneck for the scalability of the consensus model of Stellar.

\medskip

We start by defining the problem consisting of deciding whether QI holds in a given BTN.

\medskip

\noindent
\textsc{Quorum-Intersection}
\begin{myquote}
  \textbf{Input:} A BTN~$\T = \langle N, \Win \rangle$
    where the sets~$\Win(i)$ for~$i \in N$ are listed explicitly.\footnote{For the purpose of this and the following result we do not need to take into consideration the $H$ and $T_i$ elements of a BTN (Definition \ref{def:btn}). We therefore omit them for conciseness.}
  
  \textbf{Question:} Is it the case that for each two quora~$Q_1,Q_2$, $Q_1 \cap Q_2 \neq \emptyset$?
\end{myquote}

\begin{theorem}
\label{prop:qi-conp-complete}
\textsc{Quorum-Intersection} is coNP-complete.
\end{theorem}



The intractability result of Theorem~\ref{prop:qi-conp-complete} says that it may
be computationally hard, in practice, to check QI.
Such a result is  robust in the sense that the related problem of checking whether QI holds after 
the insertion of one new slide  by a node into a system that already satisfies QI,  is also coNP-complete.
(actually coNP-hard).\\

\noindent
\textsc{Slice-Addition-Quorum-Intersection}
\begin{myquote}
  \textbf{Input:} Two BTNs~$\T = \langle N, \Win \rangle$
    and~$\T' = \langle N, \Win' \rangle$,
    such that~$\T'$ satisfies QI,
    and such that~$\T$ is obtained from~$\T'$ by
    adding one single slice to~$\Win'(i)$ for some~$i \in N$,
    where the sets~$\Win(i)$ and~$\Win'(i)$
    for all~$i \in N$ are listed explicitly.
  
  \textbf{Question:} Is it the case that for each two quora~$Q_1,Q_2$ of $\T$~$Q_1 \cap Q_2 \neq \emptyset$?
\end{myquote}

\begin{theorem}
\label{prop:sliceadd-conp-complete}
\textsc{Slice-Addition-Quorum-Intersection} is coNP-complete.
\end{theorem}


\section{Quantifying Influence on Consensus in BTNs}
\label{sec:influence}

Theorem \ref{th:trust} showed that, in uniform QBTNs, safety implies the existence of nodes that are trusted by all honest nodes. While this can definitely be interpreted as a high level of centralisation required by safety, it is worth trying to precisely quantify the effect that the existence of all-trusted nodes has on consensus. In PoW and PoS protocols it is straightforward, at least by first approximation, to understand what the influence of each node is on the consensus process: each node will be able to determine a fraction of blocks corresponding to the node's share of total hashing power (PoW) or of total stakes (PoS). For consensus based on voting on trust structures, like in Ripple and Stellar, quantifying nodes' influence in a principled way is not obvious.  This section proposes a methodology for such quantification that leverages the theory of voting games. 

\subsection{Influence Matrices}

Within a BTN, a simple game is associated to each honest node. The {\em Penrose-Banzhaf index} \cite{penrose46elementary,banzhaf65weighted} of $j$ in the simple game $\tuple{T_i, \Win_i}$ of node $i$ is
\begin{align}
\B_i(j) & = \frac{1}{2^{n-1}} \sum_{C \subseteq N \backslash \set{j}} v(C \cup \set{j}) - v(C), \label{eq:b}
\end{align}
which measures the probability that node $j \in T_i$ is pivotal in a coalition to determine $i$'s opinion, assuming all other agents in $T_i$ have uniformly random opinions. Indeed $v(C) = 1$ if $C$ allows a decision to be made, and $v(C)=0$ otherwise. Informally speaking, the sum counts all the cases in which $j$ is needed to enable the decision of $i$, normalised on the possible $2^{n-1}$ coalitions without $j$.

For any agent $j \not\in T_i$ we stipulate $\B(j) = 0$, as nodes that $i$ does not trust cannot influence $i$'s opinion directly (they are `dummy agents' in the game-theoretic jargon). Byzantine nodes are assigned degenerate simple games containing a singleton winning coalition containing themselves (cf. Remark~\ref{rem:dictators} above). Byzantine agents cannot be influenced: for all  $i \in \N \backslash H$ from the degenerate simple game associated to $i$ we have $\B_i(i)=1$, and $\B_i(j)=0$ for each $j \neq i$. The normalised version of the Penrose-Banzhaf index $\NB_i(j)$  is:
\begin{align}
\NB_i(j) & = \frac{\B_i(j)}{\sum_{k \in N} \B_i(k)} \label{eq:bn}
\end{align}
Given a BTN, we associate to each honest node $i$ a vector  $[\NB_i(1), \ldots, \NB_i(n) ]$ of normalized Penrose-Banzhaf indices capturing the influence that each node has on $i$. Clearly $\sum_{j \in N} \NB_i(j) = 1$ and $\NB(j) > 0$ only if $j \in T_i$, for any $i \in N$. Notice that the vector of a Byzantine node $i$ is degenerate: $\NB_i(i)=1$ and $\NB_i(j)=0$ for each $j \neq i$.

\medskip

It follows that each BTN $\T$ can be described by a stochastic $N \times N$ matrix
\begin{align*}
I(\T) & = 
\begin{bmatrix}
    \I_{11} & \I_{12} & \I_{13} & \dots  & \I_{1n} \\
   \I_{21} & \I_{22} & \I_{23} & \dots  & \I_{2n} \\
    \vdots & \vdots & \vdots & \ddots & \vdots \\
    \I_{n1} & \I_{n2} & \I_{n3} & \dots  & \I_{nn}
\end{bmatrix}
\end{align*}
where $\I_{ij}$ denotes the normalized Penrose-Banzhaf index $\NB_i(j)$ of node $j$ in the simple game associated to $i$. We call such matrix $\I(\T) = [\I_{ij}]_{i,j \in N}$ the {\em influence matrix} (of $\T$). We will drop reference to $\T$ when not needed.
The matrix encodes the influence that each node has on each other. Matrices of this type have a long history in the mathematical modeling of influence in economics and the social sciences dating back to \cite{french56formal,DeGroot}, and have recently received renewed attention \cite{jackson08social}.\footnote{See also \cite{proskurnikov17tutorial,proskurnikov18tempo} for an overview of such models.} Similar matrices, but based on the Shapley-Shubik power index \cite{shapley54method} instead of the Penrose-Banzhaf one, have more recently been studied in \cite{hu03authority,hu03authorityB}.\footnote{For a comparison between these two power indeces we refer the reader to \cite{felsenthal98measurement}.}

\begin{example}\label{ex:influence}
Consider the following BTN with no Byzantine nodes and consisting of 6 agents all having a same set of 5 agents as trust set: $N = H = \set{1,2,3,4,5,6}$, $T_i = \set{1,2,3,4,5}$ for all $i \in N$ and $q_i = 0.8$ for all $i \in N$. By \eqref{eq:b} and \eqref{eq:bn}  for each $i \in \set{1,2,3,4,5, 6}$ we have $\B_j(i) = \frac{2}{8}$ and $\NB_j(i) = \frac{1}{5}$, for each node $j \in \set{1,2,3,4,5}$. The influence matrix describing this BTN consists of 6 identical row vectors
$
\begin{bmatrix}
    \frac{1}{5} &  \frac{1}{5} &  \frac{1}{5} &  \frac{1}{5}  &  \frac{1}{5} & 0.
\end{bmatrix}
$

Consider now a variant of the above BTN where node $5$ is Byzantine. The influence matrix describing this variant consists of 5 identical vectors 
$
\begin{bmatrix}
    \frac{1}{5} &  \frac{1}{5} &  \frac{1}{5} &  \frac{1}{5}  &  \frac{1}{5} & 0.
\end{bmatrix}
$
for the rows corresponding to nodes $1-4$ and $6$, and the degenerate row vector
$
\begin{bmatrix}
    0 &  0 &  0 &  0  & 1 & 0.
\end{bmatrix}
$
for the row of node $5$. That is, all the nodes in $T_i$ have the same influence on honest nodes, but no honest node influences $5$. 
\end{example}

\subsection{Limit Influence}

Pushing further the game theoretic framework, by using an influence matrix we can pinpoint the influence of any node $j$ on determining node $i$'s opinion: 
\begin{description}
\item{$\I_{ij}$} represents the probability that $j$ can directly sway $i$ to validate a value $x$. This is $j$'s direct influence on $i$.

\item{$\I^2_{ij}$} represents the probability that $j$ can sway $i$'s opinion in two steps, by swaying the opinion of an intermediate node $k$ which in turn sways $i$'s opinion directly. This is $j$'s indirect ($2$-step) influence on $i$.

\item{$\I^k_{ij}$} more generally represents $j$'s indirect ($k$-step) influence on $i$. 
\end{description}
So the influence (direct or indirect) of $j$ on $i$ in a BTN is given by the total probability of all ways in which $j$ can sway $i$'s opinion. Formally this amounts to
$
\lim_{t \to \infty} (\I^t)_{ij}, 
$
provided such limit exists. In yet other words, this denotes the likelihood that $j$ is able to determine the value $i$ validates. 

We are then in the position to quantify what the influence is of each node on every other node by taking the limit of the power of the influence matrix of the BTN $\T$, that is:
\begin{align}
\I(\T)^\infty & = \lim_{t \to \infty} \I(\T)^t \label{eq:limit}
\end{align}
If the limit matrix in \eqref{eq:limit} exists, we say that the influence matrix $\I(\T)$ is {\em regular}. We say that it is {\em fully regular} when its limit matrix exists and it is such that all rows are identical.\footnote{The `regularity' and `full regularity' terminology are borrowed from \cite{gantmacher00theory} and \cite{proskurnikov17tutorial}.} Intuitively, regularity means that it is possible to precisely quantify the influence of each node on each other node; full regularity means that every node has the same influence on every other node.

\begin{example}
Consider again the two BTNs introduced in Example \ref{ex:influence}. In the first case, where all nodes are honest, all nodes belonging to some trust set have positive and---given the symmetry built in the example---the same influence:
\[
\begin{bmatrix}
    \frac{1}{5} &  \frac{1}{5} &  \frac{1}{5} &  \frac{1}{5}  &  \frac{1}{5} & 0 \\
    \frac{1}{5} &  \frac{1}{5} &  \frac{1}{5} &  \frac{1}{5}  &  \frac{1}{5} & 0 \\
    \frac{1}{5} &  \frac{1}{5} &  \frac{1}{5} &  \frac{1}{5}  &  \frac{1}{5} & 0 \\
    \frac{1}{5} &  \frac{1}{5} &  \frac{1}{5} &  \frac{1}{5}  &  \frac{1}{5} & 0 \\
    \frac{1}{5} &  \frac{1}{5} &  \frac{1}{5} &  \frac{1}{5}  &  \frac{1}{5} & 0 \\
    \frac{1}{5} &  \frac{1}{5} &  \frac{1}{5} &  \frac{1}{5}  &  \frac{1}{5} & 0
\end{bmatrix}
=
\begin{bmatrix}
    \frac{1}{5} &  \frac{1}{5} &  \frac{1}{5} &  \frac{1}{5}  &  \frac{1}{5} & 0 \\
    \frac{1}{5} &  \frac{1}{5} &  \frac{1}{5} &  \frac{1}{5}  &  \frac{1}{5} & 0 \\
    \frac{1}{5} &  \frac{1}{5} &  \frac{1}{5} &  \frac{1}{5}  &  \frac{1}{5} & 0 \\
    \frac{1}{5} &  \frac{1}{5} &  \frac{1}{5} &  \frac{1}{5}  &  \frac{1}{5} & 0 \\
    \frac{1}{5} &  \frac{1}{5} &  \frac{1}{5} &  \frac{1}{5}  &  \frac{1}{5} & 0 \\
    \frac{1}{5} &  \frac{1}{5} &  \frac{1}{5} &  \frac{1}{5}  &  \frac{1}{5} & 0
\end{bmatrix}^2
=
\lim_{t \to \infty}
    \begin{bmatrix}
    \frac{1}{5} &  \frac{1}{5} &  \frac{1}{5} &  \frac{1}{5}  &  \frac{1}{5} & 0 \\
    \frac{1}{5} &  \frac{1}{5} &  \frac{1}{5} &  \frac{1}{5}  &  \frac{1}{5} & 0 \\
    \frac{1}{5} &  \frac{1}{5} &  \frac{1}{5} &  \frac{1}{5}  &  \frac{1}{5} & 0 \\
    \frac{1}{5} &  \frac{1}{5} &  \frac{1}{5} &  \frac{1}{5}  &  \frac{1}{5} & 0 \\
    \frac{1}{5} &  \frac{1}{5} &  \frac{1}{5} &  \frac{1}{5}  &  \frac{1}{5} & 0 \\
    \frac{1}{5} &  \frac{1}{5} &  \frac{1}{5} &  \frac{1}{5}  &  \frac{1}{5} & 0
\end{bmatrix}^t 
\]
In the second case, where node $5$ is Byzantine, the only node having positive influence (total influence $1$ in this example) is precisely $5$:
\[
\lim_{t \to \infty}
    \begin{bmatrix}
    \frac{1}{5} &  \frac{1}{5} &  \frac{1}{5} &  \frac{1}{5}  &  \frac{1}{5} & 0 \\
    \frac{1}{5} &  \frac{1}{5} &  \frac{1}{5} &  \frac{1}{5}  &  \frac{1}{5} & 0 \\
    \frac{1}{5} &  \frac{1}{5} &  \frac{1}{5} &  \frac{1}{5}  &  \frac{1}{5} & 0 \\
    \frac{1}{5} &  \frac{1}{5} &  \frac{1}{5} &  \frac{1}{5}  &  \frac{1}{5} & 0 \\
    0 & 0 & 0 & 0 & 1 & 0 \\
    \frac{1}{5} &  \frac{1}{5} &  \frac{1}{5} &  \frac{1}{5}  &  \frac{1}{5} & 0
\end{bmatrix}^t 
= 
\begin{bmatrix}
    0 & 0 & 0 & 0 & 1 & 0 \\
    0 & 0 & 0 & 0 & 1 & 0 \\
    0 & 0 & 0 & 0 & 1 & 0 \\
    0 & 0 & 0 & 0 & 1 & 0 \\
    0 & 0 & 0 & 0 & 1 & 0 \\
    0 & 0 & 0 & 0 & 1 & 0 
\end{bmatrix}
\]
In other words, the only node having influence on which values will be validated by other nodes, and therefore which values will be agreed upon, is the Byzantine node.
\end{example}


\subsection{Limit Influence in Ripple and Stellar}

Theorem \ref{th:trust} established that in uniform QBTN, and therefore Ripple, safety requires centralisation in the sense of requiring a non-empty set of nodes trusted by all other nodes. While this does not apply in general to Stellar, recent studies have highlighted that Stellar enjoys a similar level of centralisation.\footnote{Data analysis of the current Stellar network has shown \cite{kim2019IEEEPSB} that one of the three Stellar foundations validators is included in all trust sets. If we treat the Stellar foundation to be operating as one node, Stellar satisfies {\em de facto} the same level of centralisation that we have shown is analytically required for Ripple. }

Here we put the above methodology at work to study limit influence in centralised BTNs, that is BTNs where nodes exist that are trusted by all nodes. We show (Theorem \ref{prop:Ireg}) that: the existence of nodes trusted by all nodes makes it possible to establish limit influence (first claim); this limit influence is such that every node has the same limit influence on every other node (second claim) when at most one Byzantine node exists in the BTN; but if only just one all-trusted node trusts a Byzantine node, no honest node has limit influence on any other honest node (third claim). That is, in a centralised BTN the power of determining consensus values is all in the hands of Byzantine nodes.

\begin{theorem}
\label{prop:Ireg}
Let $\T$ be a BTN. If $\T$ is such that $\bigcap_{i \in H} T_i \neq \emptyset$ then:
\begin{enumerate}[1)]
\item $\I(\T)$ is regular;
\item $\I(\T)$ is fully regular if in addition $\T$ is such that $|B| \leq 1$;
\item and, if there exists $j \in \left( \bigcap_{i \in H} T_i \right) \cap H$ such that $T_j \cap B \neq \emptyset$ then for all $j,k \in H$, $\I(\T)_{jk}^\infty = 0$.  
\end{enumerate}
\end{theorem}
Again, it is worth noticing that this is a general protocol-independent result: it concerns all protocols working on centralized BTNs. In particular, it applies to the setup of the Ripple trust network under the assumption of safety (by Theorem \ref{th:trust}) and to the current setup of the Stellar trust network.

\label{sect:examples}

%
%
%


\section{Conclusions}
\label{sect:concl}

We have presented a framework  to quantitatively characterise decentralisation, a foundational 
and highly innovative property of blockchain technologies. Although largely discussed, decentralisation is hard to define, 
as it is a complex property depending on many aspects of the  multidisciplinary and multi-layered design of blockchains. As a consequence, 
it is also a property difficult to be formally defined and  quantitatively analysed.


We have addressed decentralisation in the specific context of BFT consensus based on open quorum systems, showcasing the relevance of tools from economic theory (command games, power indices) and computational complexity theory.
We argue that the obtained results show this is a promising general approach to the formal analysis of decentralisation. 

We focused on a general class of consensus, linking decentralisation to a precise measure of the influence of 
each peer in the network (a theme largely studied in economics), an analysis of the structural properties of the consensus network, 
and the computational complexity of some proposed solutions. The obtained limiting results on Ripple and Stellar are coherent with the current practice and the proposals that industry is putting forward to improve decentralisation.

\medskip

Our results point to several avenues of future research. 
We are planning to extend our analysis to other blockchains based on BFT consensus that are currently being developed, noticeably  Cobalt \cite{macbrough18cobalt} as an evolution of the Ripple/Stellar tradition. These will offer interesting use cases for benchmarking our approach.
More generally, we also want to explore the applicability of the methodology beyond the framework of Byzantine Trust Networks, since measures of the relative influence of peers are of interest for other blockchain frameworks, e.g. PoS. At the same time, we also intend to build on such measures to address the relationships between influence, decentralisation and, crucially, revenue. Properly understanding such mechanisms will serve to the long-term goal of designing more reliable and robust blockchains.

On the application side, the development of a prototype analysis toolkit and collection of relevant data is also an ongoing activity.

\newpage
\appendix
\section{Proofs}

\subsection{Proofs of Section~\ref{sec:prelim}}

\subsubsection{Lemma~\ref{lemma:most}}

\begin{proof}
\fbox{\eqref{uno}}
Assume $|T_i^\O(x)| > 0$. Under this assumption $j$ observes at least $|T_i^\O(x) \cap T_j \cap H|$ nodes with opinion $x$ in $T_j$. Those are the honest nodes among the nodes with opinion $x$ that both $i$ and $j$ can observe. So 
$$
|T_j^\O(x) \cap H| \geq |T_i^\O(x) \cap T_j \cap H|.
$$
Now among the nodes in $T_i^\O(x) \cap T_j$ there are at most $\beta_{ij}$ Byzantine nodes that could reveal the opposite opinion $\overline{x}$ to j. So,
$$
|T_i^\O(x) \cap T_j \cap H| \geq |T_i^\O(x) \cap T_j| - \beta_{ij}.
$$
The claim is finally established by the following series of (in)equalities:
\begin{align*}
|T_i^\O(x) \cap T_j \cap H| & \geq |T_i^\O(x) \cap T_j| - \beta_{ij} \\
	& \geq |T_i^\O(x)| - |T_i \backslash T_j| - \beta_{ij} \\
	& = |T_i^\O(x)| - |T_i| + |T_i \cap T_j| - \beta_{ij}
\end{align*}

\fbox{\eqref{due}} 
Assume $|T_i^\O(x)| > 0$. By \eqref{uno}, $|T_j^\O(x) \cap H| = |T_i \cap T_j| + |T_i^\O(x)| - |T_i| - \beta_{ij}$  whenever only the honest nodes in $T_j$ have opinion $x$. It follows that  
\begin{align*}
|T_j^\O(\overline{x})| & \leq |T_j|  -  (|T_i \cap T_j| + |T_i^\O(x)| - |T_i| - \beta_{ij}) \\
					& =  |T_j| -  |T_i \cap T_j| - |T_i^\O(x)| + |T_i| + \beta_{ij}.
\end{align*}
This completes the proof. \qed
\end{proof}

\subsubsection{Lemma~\ref{lemma:half}}

\begin{proof}
The proof consists of two sub-arguments. \fbox{First} we show that safety implies that, for all $i,j \in H$:
\begin{align}
& |T_i \cap T_j| > b \cdot (|T_i| + |T_j|) + \beta_{ij} \label{eq:intersect}
\end{align}
By safety, if $T^\O_i(x) \geq q_i |T_i|$ with $i \in H$, then for all $j \in H$ $T^\O_j(\overline{x}) < q |T_j|$. Assume $T^\O_i(x) \geq q |T_i|$ with $i \in H$. By Lemma \ref{lemma:most} and safety we have:
\begin{align*}
T^\O_j(\overline{x}) & \leq |T_j| -  |T_i \cap T_j| -  q|T_i| + |T_i| + \beta_{ij}  \\
				& <  q |T_j|
\end{align*}
From $|T_j| -  |T_i \cap T_j| -  q|T_i| + |T_i| + \beta_{ij}  <  q |T_j|$ we thus obtain 
\begin{align*}
|T_j| - q|T_j| -  q|T_i| + |T_i| + \beta_{ij}  & = b \cdot (|T_i| + |T_j|) + \beta_{ij} \\
& <  |T_i \cap T_j|
\end{align*}
as desired.\footnote{Cf. \cite[Proposition 4]{chase18analysis}.}

\fbox{Second} We show that safety also implies that, for all $i, j \in H$:
\begin{align}
b \cdot (|T_i| + |T_j|) + \beta_{ij} >  \frac{b}{1-b} (|T_i| + |T_j|) \label{eq:intersectx}
\end{align}

We have established that safety implies that for all $i,j \in H$, $|T_i \cap T_j| > b \cdot (|T_i| + |T_j|) + \beta_{ij}$ \eqref{eq:intersect}, that is, the size of the intersection of the trust sets of $i$ and $j$ should be larger than the maximum possible fraction of Byzantine nodes times the combined size of the trust sets, plus $\beta_{ij}$. Now recall the definition of $\beta_{ij}$ \eqref{eq:beta}. By \eqref{eq:intersect} it cannot be the case that $\beta_{ij} = |T_i \cap T_j|$. So $\beta_{ij} = b |T_k|$ where $T_k$ is the the smallest set between $T_i$ and $T_j$. Now assume, w.l.o.g. that $|T_i| \geq |T_j|$ and so that $|T_j| = x \cdot |T_i|$ with $x \in (0,1]$. By \eqref{eq:intersect} we have:
\begin{align*}
|T_i \cap T_j| & > b (|T_i| + |T_j|) + \beta_{ij} \\
& = b (|T_i| + x |T_i|) + b x |T_i| \\
& = b |T_i| (1 + 2x) 
\end{align*}
From this, and the fact that a set is always at least as large as its intersection with another we obtain a lower bound for $x$ by the following series of inequalities:
\begin{align*}
x|T_i| & \geq |T_i \cap T_j| \\
x|T_i| & > b |T_i| (1 + 2x) \\
x & > b (1 + 2x) \\
x & > b + 2bx \\
x - 2bx & > b \\
x(1-2b) & > b \\
x & > \frac{b}{1-2b}
\end{align*}
By substituting $\frac{b}{1-2b}$ for $x$ in \eqref{eq:intersect} we thus obtain a lower bound for $|T_i \cap T_j|$ in $b$. We then reformulate \eqref{eq:intersect} in terms of the combined size $\alpha = |T_i| + |T_j|$ of the two trust sets:
\begin{align*}
|T_i \cap T_j|  & >  b (\underbrace{|T_i| + x |T_i|}_\alpha) + b x |T_i| \\
& > b (|T_i| + \frac{b}{1-2b}|T_i|) + b \frac{b}{1-2b} |T_i| \\
&=  b \alpha + b \frac{b}{1-2b} \frac{1-2b}{1 - b} \alpha \ \ \ \ \ \ \ \ \ \ \ \ \ \ \ \ \ \ \ \ \ \ \ \ \mbox{as} |T_i| = \alpha \frac{1-2b}{1 - b} \\
& = b \alpha (1 + \frac{b}{1-b})\\
& = \frac{b}{1-b} \alpha
\end{align*} 
So safety implies that the size of the intersection of $T_i$ and $T_j$ must be larger than the fraction $\frac{b}{1-b}$ of the combined size of the two sets. \qed
\end{proof}

\subsubsection{Lemma~\ref{lemma:intersect}}

\begin{proof}
The proof is by induction on $|H|$. \fbox{Base} if $|H| = 1$ the claim holds trivially. \fbox{Step} Now assume the claim holds for $|H| = m$ (IH). We prove it holds for $|H| = m+1$. So assume for all $i,j \in H$, $|T_i \cap T_j| > 0.25 \cdot (|T_i| + |T_j|)$, and let $k$ be the $m+1^{\mbox{th}}$ node in $H$. By IH we know that $\bigcap_{i \in H \backslash \set{k}} T_i \neq \emptyset$. Now take one of the smallest (w.r.t. size) $H_i$ with $i \in H \backslash \set{k}$ and call it $H_j$. There are two cases. \fbox{$|H_k| \leq |H_j|$}. Then $|H_j \cap H_k| \geq 0.5 \cdot |H_k|$. Since $H_j$ was smallest amongst the $H_i$, it also hols that $\forall i \in H$ $|H_i \cap H_k| \geq 0.5 \cdot |H_k|$. From this we conclude that $\bigcap_{i \in H} T_i \neq \emptyset$. \fbox{$|H_k| > |H_j|$} Then $|H_j \cap H_k| \geq 0.5 \cdot |H_j|$. Since $H_j$ was smallest amongst the $H_i$, it also hols that $\forall i \in H$ $|H_i \cap H_j| \geq 0.5 \cdot |H_j|$, from which we also conclude $\bigcap_{i \in H} T_i \neq \emptyset$. \qed
\end{proof}


\subsubsection{Theorem \ref{th:QI}}

\begin{proof}
\fbox{Left to right} Straightforwardly proven by contraposition.
\fbox{Right to left} Proceed by contraposition and assume there is a profile $\O$, a function $s$ and agents $1$ and $2$ such that all $k \in C_1 = \bigcup_{1 \leq m} F_s^m(\set{1})$ agree on $x$ and all $k \in C_2 = \bigcup_{1 \leq m} F_s^m(\set{2})$ agree on $\overline{x}$. Observe that $C_1$ and $C_2$ are quora containing (since the BTN is vetoed) $1$ and $2$. There are two cases. Either $C_1 \cap C_2 = \emptyset$, or if that is not the case then $C_1 \cap C_2 \subseteq B$ as only Byzantine nodes can reveal different opinions to different nodes. Hence $C_1$ and $C_2$ are either disjoint or their intersection contains only Byzantine nodes. \qed
\end{proof}


\subsubsection{Theorem \ref{prop:qi-conp-complete}}

\begin{proof}
To see that the problem is contained in coNP,
we describe a nondeterministic polynomial-time algorithm
to decide whether~$\C$ does not have the quorum
intersection property.
The algorithm guesses two disjoint sets~$Q_1,Q_2 \subseteq N$.
Then, for each~$u \in [2]$ and for each~$i \in Q_u$,
the algorithm checks if there is some~$C \in \Win(i)$
such that~$C \subseteq Q_u$.
That is, the algorithm verifies that~$Q_1$ and~$Q_2$ are quora
(which is the case if and only if all checks succeed).
Clearly, all checks can be performed in polynomial time.
Thus, deciding whether~$\C$ has the quorum intersection
property is in coNP.

To show coNP-hardness, we reduce from the coNP-complete
propositional unsatisfiability problem (\textsc{UNSAT}).
Let~$\varphi$ be a propositional formula
containing the propositional variables~$x_1,\dotsc,x_n$.
Without loss of generality, we may assume that~$\varphi$ is in 3CNF,
i.e., that~$\varphi = c_1 \wedge \dotsm \wedge c_m$
and that for each~$j \in [m]$,~$c_j = (T_{j,1} \vee T_{j,2} \vee T_{j,3})$,
where~$T_{j,1},T_{j,2},T_{j,3}$ are literals.
We construct a command game~$\C = \langle N, \Win \rangle$
that has the quorum intersection property if and only
if~$\varphi$ is unsatisfiable.

We let:
\[ N = \{ z_{0}, z_{1} \} \cup \{ c_j\ |\ j \in [m] \} \cup \{ y_i, p_i, n_i\ |\ i \in [n] \}. \]
That is, we have nodes~$z_0,z_1$, a node~$c_j$ for each clause of~$\varphi$,
and nodes~$y_i,p_i,n_i$ for each variable occurring in~$\varphi$.

We define the sets of winning coalitions of the nodes in~$N$ as follows:
\[
\begin{array}{r l p{3.5cm}}
  \Win(z_0) = & \{ \{  z_0, y_1,\dotsc,y_n \} \}; & \\
  \Win(z_1) = & \{ \{  z_1, c_1,\dotsc,c_m \} \}; & \\
  \Win(y_i) = & \{ \{ y_i, p_i \}, \{ y_i, n_i \} \} & $\mathit{for each}~i \in [n]$; \\
  \Win(c_j) = & \{ \{ c_j, \sigma(T_{j,1}) \}, \{ c_j, \sigma(T_{j,2}) \}, \{ c_j, \sigma(T_{j,3}) \} \} & $\mathit{for each}~j \in [m]$; \\
  \Win(p_i) = & \{ \{ p_i, z_0 \}, \{ p_i, z_1 \}  \} & $\mathit{for each}~i \in [n]$; $\mathit{and}$ \\
  \Win(n_i) = & \{ \{ n_i, z_0 \}, \{ n_i, z_1 \}  \} & $\mathit{for each}~i \in [n]$; \\
\end{array} 
\]
where for each positive literal~$x_i$, we let~$\sigma(x_i) = p_i$;
and for each negative literal~$\neg x_i$, we let~$\sigma(\neg x_i) = n_i$.

We argue that~$\C = \langle N, \Win \rangle$
has the quorum intersection property if and only
if~$\varphi$ is unsatisfiable.
We show the equivalent statement
that~$\C = \langle N, \Win \rangle$
does \textbf{not} have the quorum intersection property if and only
if~$\varphi$ is \textbf{satisfiable}.

$(\Rightarrow)$
Suppose that there exist~$Q_1,Q_2 \in \bm{\mathrm{Q}}^{\C}$
such that~$Q_1 \cap Q_2 = \emptyset$.
We may assume without loss of generality that~$Q_1$ and~$Q_2$
are core quora.
We know that neither~$Q_1$ nor~$Q_2$ can contain both~$z_0$ and~$z_1$,
because each quorum of~$\C$ must contain either~$z_0$ or~$z_1$
(by the specific construction of~$\C$).
Thus, we may assume that~$z_0 \in Q_2$ and~$z_1 \in Q_1$.

Then also~$\{ y_1,\dotsc,y_n \} \subseteq Q_2$.
Moreover, for each~$i \in [n]$, we know that then either~$p_i \in Q_2$
or~$n_i \in Q_2$ (and not both).
We also know that~$\{ c_1,\dotsc,c_m \} \subseteq Q_1$.
Now define the truth assignment~$\alpha : \{ x_1,\dotsc,x_n \} \rightarrow \{ 0,1 \}$
as follows.
For each~$i \in [n]$, we let~$\alpha(x_i) = 1$ if~$n_i \in Q_2$
and we let~$\alpha(x_i) = 0$ if~$p_i \in Q_2$.

We show that~$\alpha$ satisfies~$\varphi$.
Take an arbitrary clause~$c_j$ of~$\varphi$.
Due to the construction of~$\Win(c_j)$,
we know that~$Q_1$ contains (at least)
one of~$\sigma(T_{j,1}),\sigma(T_{j,2}),\sigma(T_{j,3})$.
Take some~$u \in [3]$ such that~$\sigma(T_{j,u}) \in Q_1$.
We show that~$\alpha$ satisfies~$T_{j,u}$.
To derive a contradiction, suppose the contrary,
i.e., that~$\alpha$ does not satisfy~$T_{j,u}$.
Then~$\sigma(T_{j,u}) \in Q_2$
(by the construction of~$\alpha$),
and thus~$Q_1 \cap Q_2 \neq \emptyset$,
which contradicts our initial assumption that~$Q_1 \cap Q_2 = \emptyset$.
Thus, we can conclude that~$\alpha$ satisfies~$T_{j,u}$.
This concludes our proof that~$\varphi$
is satisfiable.

$(\Leftarrow)$
Conversely, suppose that~$\varphi$ is satisfiable, i.e.,
that there is some truth assignment~$\alpha : \{ x_1,\dotsc,x_n \} \rightarrow \{ 0,1 \}$
that satisfies all clauses of~$\varphi$.
For each clause~$c_j$, define~$\rho(c_j)$ to be some literal~$T_{j,u}$ in~$c_j$
that is satisfied by~$\alpha$.
Moreover, for each~$i \in [n]$, let~$\mu(x_i) = n_i$ if~$\alpha(x_i) = 1$ and~$\mu(x_i) = p_i$
if~$\alpha(x_i) = 0$.
Consider the following two sets~$Q_1,Q_2 \subseteq N$:
\[ \begin{array}{r l}
  Q_1 = & \{ z_1,c_1,\dotsc,c_m \} \cup \{ \sigma(\rho(c_j))\ |\ j \in [m] \}; \text{ and} \\
  Q_2 = & \{ z_0,y_1,\dotsc,y_n  \} \cup \{ \mu(x_i)\ |\ i \in [n]  \}; \\
\end{array} \]
It is straightforward to verify that~$Q_1$ and~$Q_2$ are both quora,
i.e., that~$Q_1,Q_2 \in \bm{\mathrm{Q}}^{\C}$.
Moreover, since it holds that~$Q_1 \cap Q_2 = \emptyset$,
we know that~$\C$ does not satisfy the quorum intersection property. \qed
\end{proof}

\subsubsection{Theorem \ref{prop:sliceadd-conp-complete}}
\begin{proof}[sketch]
Membership in coNP follows directly from
Proposition~\ref{prop:qi-conp-complete}.
We show coNP-hardness by modifying the reduction
given in the proof of Proposition~\ref{prop:qi-conp-complete}.
We describe a reduction from UNSAT.
Let~$\varphi$ be a propositional formula
containing the propositional variables~$x_1,\dotsc,x_n$.
Without loss of generality, we may assume that~$\varphi$ is in 3CNF.
Moreover, without loss of generality, we may assume
that~$\varphi[x_1 \mapsto 1]$ is unsatisfiable.
We construct the command game~$\C =
\langle N, \Win \rangle$ as in the proof of
Proposition~\ref{prop:qi-conp-complete}.
Moreover, we transform~$\C$ into~$\C'$
by removing the slice~$\{ y_1, n_1 \}$ from~$\Win(y_1)$.
By a similar argument as the one used in the proof
of Proposition~\ref{prop:qi-conp-complete}, we know that
the command game~$\C'$ satisfies the quorum
intersection property, because~$\varphi[x_1 \mapsto 1]$
is unsatisfiable.
Moreover,~$\C$ satisfies the quorum intersection property
if and only if~$\varphi$ is unsatisfiable. \qed
\end{proof}

\subsection{Proofs of Section~\ref{sec:influence}}

\subsubsection{Theorem \ref{prop:Ireg}}

\begin{proof}
We first need to introduce some auxiliary notation. Given an influence matrix $\I$, $\g(\I) = \tuple{N, E}$ denotes the (directed) graph of $\I$, where $E = \set{ij \mid I_{ji} > 0}$. Intuitively $ji \in E$ whenever $j$ influences $i$ (i.e., has a positive Banzhaf-Penrose index in $i$'s simple game). By assumption there exist nodes that influence all other honest nodes (themselves included). Let now $E(\bigcap_{i \in H} T_i) = \set{i \in N \mid \exists j \in \bigcap_{i \in H} T_i, ij \in E}$, that is, the set of nodes that influence some node that influences all honest nodes. We distinguish three cases.

\fbox{$|B| = 0$} So there are no Byzantine nodes in $\T$, and $N = H = E(\bigcap_{i \in H} T_i)$
It follows that $\g(\I(\T))$ is strongly connected (there exists a path from every node to every node) and aperiodic (there are no two cycles in the graph which are divided by an integer larger than $1$). Trivially, it is also closed (there exists no node outsie $\g(\I(\T))$ that influences nodes in $\g(\I(\T))$).
The full regularity of $\I$ then follows from known results on influence matrices (cf. \cite[Lemma 11]{proskurnikov17tutorial}): if $\g(\I)$ contains only one strongly connected component and is aperiodic, then $\I$ is fully regular.

\fbox{$|B| = 1$ and $E(\bigcap_{i \in H} T_i) \cap B \neq \emptyset$} So there exists exactly one Byzantine node in $N$, which furthermore belongs to $E(\bigcap_{i \in H} T_i)$ (that is, it influences at least one honest node influencing in turn all honest nodes). Call $i$ such Byzantine agent. Recall that, by construction, $ii\in E$. So the subgraph consisting of $i$ and the self-loop $ii$ is the only closed, aperiodic, strongly connected component of $\g(\I(\T))$. Full regularity therefore follows by know results as in the previous case.

\fbox{$|B| \geq 1$ or $E(\bigcap_{i \in H} T_i) \cap B = \emptyset$} So there exist several Byzantine nodes in $N$ or there are Byzantine nodes which do not influence nodes in $\bigcap_{i \in H} T_i$. Both such cases determine the existence, by arguments analogous to those provided for the previous two cases, of several closed, aperiodic strongly connected components in $\g(\I(\T))$. The regularity of $\I$ then follows again from known results on influence matrices (cf. \cite[Theorem 12]{proskurnikov17tutorial}, \cite[Theorem 8.1]{jackson08social}: if all closed strongly connected components of $\g(\I)$ are aperiodic, then $\I$ is regular.

In all three cases $\I$ is regular, proving claim 1). In the first two cases ($|B| \leq 1$) $\I$ is furthermore fully regular, establishing claim 2). Finally, to prove claim 3) we reason as follows. If there exists a honest agent in $\bigcap_{i \in H} T_i$ trusting a Byzantine agent, then the only closed strongly connected components of $\g(\I)$ are the Byzantine nodes. In the limit, such nodes will therefore be the only ones having positive influence.
\qed
\end{proof}

\bibliographystyle{plain}
\bibliography{biblio.bib,addenda.bib}

\begin{thebibliography}{10}

\bibitem{banzhaf65weighted}
J.~Banzhaf.
\newblock Weighted voting doesn't work: A mathematical analysis.
\newblock {\em Rutgeres Law Review}, 19:317--343, 1965.

\bibitem{bonneau2015sok}
Joseph Bonneau, Andrew Miller, Jeremy Clark, Arvind Narayanan, Joshua~A Kroll,
  and Edward~W Felten.
\newblock Sok: Research perspectives and challenges for bitcoin and
  cryptocurrencies.
\newblock In {\em Security and Privacy (SP), 2015 IEEE Symposium on}, pages
  104--121. IEEE, 2015.

\bibitem{comsoc_handbook}
F.~Brandt, V.~Conitzer, U.~Endriss, J.~Lang, and A.~Procaccia, editors.
\newblock {\em Handbook of Computational Social Choice}.
\newblock Cambridge University Press, 2016.

\bibitem{cachin17blockchain}
C.~Cachin and M.~Vukolic.
\newblock Blockchain consensus protocols in the wild.
\newblock Technical report, CoRR abs/1707.01873, 2017.

\bibitem{Castro:1999:PBFT}
Miguel Castro and Barbara Liskov.
\newblock Practical byzantine fault tolerance.
\newblock In {\em Proceedings of the Third Symposium on Operating Systems
  Design and Implementation}, OSDI '99, pages 173--186, Berkeley, CA, USA,
  1999. USENIX Association.

\bibitem{chase18analysis}
B.~Chase and MacBrough E.
\newblock Analysis of the {XRP} ledger consensus protocol.
\newblock Technical report, Ripple Research, 2018.

\bibitem{davey90introduction}
B.~A. Davey and H.~A. Priestley.
\newblock {\em Introduction to Lattices and Order}.
\newblock Cambridge University Press, 1990.

\bibitem{DeGroot}
Morris~H. De{G}root.
\newblock {Reaching a Consensus}.
\newblock {\em Journal of the American Statistical Association},
  69(345):118--121, 1974.

\bibitem{PoW93}
Cynthia Dwork and Moni Naor.
\newblock Pricing via processing or combatting junk mail.
\newblock In Ernest~F. Brickell, editor, {\em Advances in Cryptology ---
  CRYPTO' 92}, pages 139--147, Berlin, Heidelberg, 1993. Springer Berlin
  Heidelberg.

\bibitem{macbrough18cobalt}
MacBrough E.
\newblock Cobalt: {BFT} governance in open networks.
\newblock Technical report, Ripple Research, 2018.

\bibitem{Eyal2018majority}
Ittay Eyal and Emin~G\"{u}n Sirer.
\newblock Majority is not enough: Bitcoin mining is vulnerable.
\newblock {\em Commun. ACM}, 61(7):95--102, June 2018.

\bibitem{felsenthal98measurement}
D.~Felsenthal and M.~Machover.
\newblock {\em The Measurement of Voting Power}.
\newblock Edward Elgar Publishing, 1998.

\bibitem{Fischer1985impossibility}
Michael~J. Fischer, Nancy~A. Lynch, and Michael~S. Paterson.
\newblock Impossibility of distributed consensus with one faulty process.
\newblock {\em J. ACM}, 32(2):374--382, April 1985.

\bibitem{french56formal}
J.~French.
\newblock A formal theory of social power.
\newblock {\em Psychological Review}, 61:181--194, 1956.

\bibitem{gantmacher00theory}
F.~Gantmacher.
\newblock {\em The Theory of Matrices}.
\newblock AMS Chealsea Publishing, 1959.

\bibitem{Garay2015}
J.~Garay, A.~Kiayias, and Nikos L.
\newblock {\em The Bitcoin Backbone Protocol: Analysis and Applications}.
\newblock Springer, 2015.

\bibitem{garcia-molina85how}
Hector Garcia-Molina and Daniel Barbara.
\newblock How to assign votes in a distributed system.
\newblock {\em J. ACM}, 32(4):841--860, October 1985.

\bibitem{Gervais2014BTCcentralised}
A.~Gervais, G.~O. Karame, V.~Capkun, and S.~Capkun.
\newblock Is bitcoin a decentralized currency?
\newblock {\em IEEE Security and Privacy}, 12(03):54--60, may 2014.

\bibitem{ghosh07mechanism}
A.~Ghosh, M.~Mahdian, D.~Reeves, D.~Pennock, and R.~Fugger.
\newblock Mechanism design on trust networks.
\newblock In {\em Proceedings of the 3rd International Workshop On Internet And
  Network Economics (WINE 2007)}, number 4858 in LNCS. Springer, 2007.

\bibitem{grandi18voting}
U.~Grandi.
\newblock Social choice on social networks.
\newblock In U.~Endriss, editor, {\em Trends in Computational Social Choice},
  pages 169--184. COST, 2018.

\bibitem{hu03authorityB}
X.~Hu and L.~Shapley.
\newblock On authority distributions in organizations: Controls.
\newblock {\em Games and Economic Behavior}, 45:153--170, 2003.

\bibitem{hu03authority}
X.~Hu and L.~Shapley.
\newblock On authority distributions in organizations: Equilibrium.
\newblock {\em Games and Economic Behavior}, 45:132--152, 2003.

\bibitem{jackson08social}
Matthew~O. Jackson.
\newblock {\em {Social and Economic Networks}}.
\newblock Princeton University Press, Princeton, NJ, USA, 2008.

\bibitem{karos15indirect}
D.~Karos and H.~Peters.
\newblock Indirect control and power in mutual control structures.
\newblock {\em Games and Economic Behavior}, 92, 2015.

\bibitem{kim2019IEEEPSB}
Minjeong Kim, Yujin Kwon, and Yongdae Kim.
\newblock Is stellar as secure as you think?
\newblock In {\em IEEE Security and Privacy on the Blockchain (IEEE S\&B '19)}.
  IEEE, 2019.

\bibitem{lamport82byzantine}
L.~Lamport, R.~Shostak, and M.~Pease.
\newblock The {Byzantine} generals problem.
\newblock {\em ACM Transactions on Programming Languages and Systems},
  4(3):382--401, 1982.

\bibitem{Lachowski2019QInp}
{L$\!\!\!$\mbox{-}}ukasz Lachowski.
\newblock Complexity of the quorum intersection property of the federated
  byzantine agreement system.
\newblock \url{https://arxiv.org/abs/1902.06493}, 2019.

\bibitem{malkhi98byzantine}
D.~Malkhi and M.~Reiter.
\newblock Byzantine quorum systems.
\newblock {\em Distributed Computing}, 11:203--213, 1998.

\bibitem{lokhava19fast}
David Mazières Graydon Hoare Nicolas Barry Eliezer Gafni Jonathan Jove Rafał
  Malinowsky Jed~McCaleb Marta~Lokhava, Giuliano~Losa.
\newblock Fast and secure global payments with {Stellar}.
\newblock In {\em Proceedings of {SOSP}'19}. ACM, 2019.

\bibitem{mazieres16stellar}
D.~Mazier\`es.
\newblock The stellar consensus protocol: A federated model for internet-level
  consensus.
\newblock Stellar Development Foundation, 2016.

\bibitem{bitcoin}
Satoshi Nakamoto.
\newblock Bitcoin: A peer-to-peer electronic cash system.
\newblock {\em Bitcoin project white paper}, 2009.

\bibitem{narayanan16handbook}
A.~Narayanan, J.~Bonneau, E.~Felten, A.~Miller, and S.~Goldfeder.
\newblock {\em Bitcoin and Cryptocurrency Technologies}.
\newblock Princeton University Press, 2016.

\bibitem{penrose46elementary}
L.~Penrose.
\newblock The elementary statistics of majority voting.
\newblock {\em Journal of the Royal Statistical Society}, 109(1):53--57, 1946.

\bibitem{proskurnikov17tutorial}
A.~Proskurnikov and R.~Tempo.
\newblock A tutorial on modeling and analysis of dynamic social networks. {Part
  I}.
\newblock {\em Annual Reviews in Control}, 43:65--79, 2017.

\bibitem{proskurnikov18tempo}
A.~Proskurnikov and R.~Tempo.
\newblock A tutorial on modeling and analysis of dynamic social networks. {Part
  II}.
\newblock {\em Annual Reviews in Control}, 45:166--190, 2018.

\bibitem{schneider90implementing}
F.~Schneider.
\newblock Implementing fault-tolerant services using the state machine
  approach: A tutorial.
\newblock {\em ACM Computing Surveys}, 22(4):299--319, 1990.

\bibitem{schwartz14ripple}
D.~Schwartz, N.~Youngs, and A.~Britto.
\newblock The ripple protocol consensus algorithm.
\newblock Technical report, Ripple Labs, 2014.

\bibitem{shapley54method}
L.~Shapley and M.~Shubik.
\newblock A method for evaluating the distribution of power in a committee
  system.
\newblock {\em American Political Science Review}, 48:787--792, 1954.

\bibitem{vukolic12quorum}
M.~Vukolic.
\newblock {\em Quorum Systems with Applications to Storage and Consensus}.
\newblock Morgan \& Claypool Publishers, 2012.

\bibitem{vukolic15quest}
M.~Vukolic.
\newblock The quest for scalable blockchain fabric: Proof-of-work vs. {BFT}
  replication.
\newblock In {\em Proceedings of {iNetSec}'15}, volume 9591 of {\em LNCS},
  pages 112--125, 2015.

\bibitem{wattenhofer17distributed}
R.~Wattenhofer.
\newblock {\em Distributed Ledger Technology: The Science of the Blockchain}.
\newblock Createspace Independent Publishing Platform, 2017.

\bibitem{garcaprez2018OPODISStellar}
çlvaro Garc{\'i}a-P{\'e}rez and Alexey Gotsman.
\newblock {Federated Byzantine Quorum Systems}.
\newblock In J.~Cao, F.~Ellen, L.~Rodrigues, and B.~Ferreira, editors, {\em
  22nd International Conference on Principles of Distributed Systems (OPODIS
  2018)}, volume 125 of {\em LIPIcs}, pages 17:1--17:16, 2018.

\end{thebibliography}

\end{document}